\documentclass[twocolumn,5p]{elsarticle}

\usepackage{lineno}

\usepackage{amsmath,amsfonts,amssymb,amsthm}
{
      \theoremstyle{plain}
      \newtheorem{assumption}{Assumption}
      \newtheorem{theorem}{Theorem}
}
\usepackage{graphicx}
\usepackage{multicol}
\usepackage{multirow}
\usepackage{wrapfig}
\usepackage{float}
\usepackage{tabularx,ragged2e,booktabs,caption}
\usepackage{mathtools, cuted}
\usepackage{multicol,lipsum}
\usepackage{flushend}
\usepackage{array}
\usepackage{varwidth}
\newcolumntype{P}[1]{>{\centering\arraybackslash}p{#1}}
\usepackage{subcaption}
\usepackage{filecontents}
\newtheorem{remark}{Remark}
\usepackage[colorlinks=true]{hyperref}
\biboptions{compress}

\usepackage{csquotes}
\usepackage[nameinlink,capitalise]{cleveref}

\Crefname{equation}{}{}
\journal{Engineering Science and Technology}

\begin{document}

\begin{frontmatter}

\title{Nonlinear PID Controller Design for a 6-DOF UAV Quadrotor System}

\author[iq]{Aws Abdulsalam Najm}
\author[iq]{Ibraheem Kasim Ibraheem}
\ead{ibraheem.i.k@ieee.org}

\address[iq]{Baghdad University - College of Engineering, Electrical Engineering Department, 10001, Baghdad, Iraq}

\begin{abstract}
A Nonlinear PID (NLPID) controller is proposed to stabilize the translational and rotational motion of a 6-DOF UAV quadrotor system and enforce it to track a given trajectory with minimum energy and error. The complete nonlinear model of the 6-DOF quadrotor system are obtained using Euler-Newton formalism and used in the design process, taking into account the velocity and acceleration vectors resulting in a more accurate 6-DOF quadrotor model and closer to the actual system. Six NLPID controllers are designed, each for Roll, Pitch, Yaw, Altitude, and the Position subsystems, where their parameters are tuned using GA to minimize a multi-objective Output Performance Index (OPI). The stability of the 6-DOF UAV subsystems has been analyzed in the sense of Hurwitz stability theorem under certain conditions on the gains of the NLPID controllers.  The simulations have been accomplished under MATLAB/SIMULINK environment and included three different trajectories, i.e., circular, helical, and square. The proposed NLPID controller for each of the six subsystems of the 6-DOF UAV quadrotor system has been compared with the Linear PID (LPID) one and the simulations showed the effectiveness of the proposed NLPID controller in terms of speed, control energy, and steady-state error.
\end{abstract}


\end{frontmatter}

%%\linenumbers

\section{Introduction}\label{Introduction}
A quadrotor is one of the unmanned aerial vehicles (UAV), it does not need a pilot to be controlled. Quadrotor has four arms each one has a rotor on it. Every two rotors on the same axis rotate in the same direction and opposite to the other two rotors direction of rotation. The quadrotor is under-actuated system, because the number of rotors is less than the number of DOF. This makes designing a controller is a difficult problem. Quadrotors applications have been increased in the last years because of their simple implementation, low cost, different sizes, and maneuverability. Many applications founded for danger places, disasters, and rescue \cite{Kumar2015,DallalBashi2017,Saha2018}, quadrotor applications in agriculture \cite{Navia2016a,Daroya2017}, and even in helpful jobs like \cite{Hunt2014a,Camci2016}. The work in \cite{Kim2018} is a good survey for quadrotor applications for entertainment. Many studies could be found for the multi-agent system and formation control like \cite{Eskandarpour2014,Rabah2016,Ghamry2016a}, and for more applications in multi-agent systems we recommend \cite{Hou2017}. Many researchers studied the quadrotor control design with different types of controllers. One of the most used controllers is the LPID control because of its simplicity \cite{Kotarski2016,Hadi2017a,Akhil2011,Alsharif2017a} but on the other hand it has many disadvantages: 1) sometimes it gives a high control signal due to the fact of windup, thus, overshooting and continuing to increase as the accumulated error is unwound (offset by errors in the other direction, 2) the differentiator leads to noise amplification. Other used controllers like: NLPID control \cite{Gonzalez-Vazquez2010}\, and \cite{Benrezki2015}, LQR \cite{Argentim2013a}, geometric control \cite{Lee2010a}, nonlinear model predictive control \cite{Ru2017a}, L1 control \cite{Thu2017a}, fuzzy control \cite{Seidabad2014,EduardoM.Bucio-Gallardo2016}. A further studies for control algorithms with quadrotor systems can be found in \cite{Zulu2014a}.
\par In this paper, a NLPID controller proposed in our previous work \cite{Riyadh2017} which a nonlinear combinations of the error signal is used to stabilize a 6-DOF quadrotor system, its stability verified using Hurwitz stability and its performance compared with that of the most famous one, i.e., the LPID controller. The control system for the 6-DOF UAV consists of six NLPID controllers, three NLPID controllers for the translational system and the rest ones for the rotational system of the underlying UAV, with twelve tuning parameters for each NLPID controller, they are tuned using Genetic Algorithm (GA) and optimized toward the minimization of the proposed multi-objective OPI which is a weighted sum of the Integrated Time Absolute Error (ITAE) and the square of the control signal $U$ (USQR).\par
The structure of this paper is as follows: \autoref{Mathematical Model} presents modeling of the 6-DOF quadrotor system. Next, \autoref{Problem statement} describes the problem statement. The main results: the NLPID controller design and stability analysis are given in \autoref{Nonlinear Controller Design and stability Analysis}. \autoref{Simulation and Results} illustrates the numerical simulations and discussions followed by a conclusion and future work in \autoref{Conclusion and Future Work}.
\section{Mathematical Modelling of the 6-DOF UAV Quadrotor}\label{Mathematical Model}
To control any system, firstly, a mathematical model must be derived. This mathematical model will describe the responses of the system for different inputs. The inputs for the 6-DOF quadrotor system are combinations of the rotors speed ($\Omega$), which in this case is a force $f_t$ to control the altitude ($z$) and the torques ($\tau_x$, $\tau_y$, and $\tau_z$) to control the angels ($\phi$, $\theta$, and $\psi$) respectively, see \eqref{eq2}. The meaning of each parameter is described in \autoref{table1}.
\begin{equation}\label{eq2}
\left\{
\begin{array}{l}
f_t=b(\Omega_1^2+\Omega_2^2+\Omega_3^2+\Omega_4^2)\\
\tau_x=bl(\Omega_3^2-\Omega_1^2 )\\
\tau_y=bl(\Omega_4^2-\Omega_2^2 )\\
\tau_z=d(\Omega_2^2+\Omega_4^2-\Omega_1^2-\Omega_3^2 )
\end{array}
\right.
\end{equation}
\begin{center}
\captionof{table}{Parameters Description} \label{table1} 
\scalebox{0.8}{
  \begin{tabular}{ c | c | c }
    \hline
    Parameters & Description & Units \\ \hline
    $[x\,y\,z]$ & Linear position vector & $m$ \\ \hline
    $[\phi\,\theta\,\psi]$ & Angular position vector & $rad$ \\ \hline
    $[u\,v\,w]$ & Linear velocity vector & $m/sec$ \\ \hline
    $[p\,q\,r]$ & Angular velocity vector & $rad/sec$ \\ \hline
    $[I_x\,I_y\,I_z]$ & Moment of inertia vector & $kg.m^2$ \\ \hline
    $f_t$ & Total thrust generated by rotors & $N$ \\ \hline
    $[\tau_x\,\tau_y\,\tau_z]$ & Control torques & $N.m$ \\ \hline
    $[f_{wx}\,f_{wy}\,f_{wz}]$ & Wind force vector & $N$ \\ \hline
    $[\tau_{wx}\,\tau_{wy}\,\tau_{wz}]$ & Wind torque vector & $N.m$ \\ \hline
    $g$ & Gravitational force & $m/sec^2$ \\ \hline
    $m$ & Total mass & $Kg$ \\ \hline
    $[\Omega_1\,\Omega_2\,\Omega_3\,\Omega_4]$ & Rotors speeds vector & $rad/sec$ \\ \hline
    $b$ & Thrust coefficient & $N.sec^2$ \\ \hline
    $l$ & Motor to center length & $m$ \\ \hline
    $d$ & Drag coefficient & $N.m.sec^2$ \\
    \hline
  \end{tabular}
  }
  \par $c( )\equiv cos⁡( ),s( )\equiv sin⁡( ),and\,t( )\equiv tan⁡( )$
\end{center}\par
The four possible movements of the 6-DOF quadrotor are shown in \autoref{figPossible Movements}. Some researchers \cite{Kotarski2016,Hadi2017a,Akhil2011,Nugraha2017a} depend only on the equations of acceleration for the 6-DOF quadrotor system without taking the velocities into account. In this paper the mathimatical model of the 6-DOF quadrotor system is crafted in such a way that the acceleration and velocity vectors are taken into consideration resulting in a more accurate nonlinear model for the 6-DOF quadrotor system and closer to the actual one. The equations of the nonlinear mathematical model of the 6-DOF quadrotor system are derived based on Euler-Newton formaliation and written in \eqref{eq3} \cite{Sabatino2015}. \autoref{figQuadrotor Dynamics Relations} shows the block diagram of the 6-DOF quadrotor dynamical relations.\par
As said earlier, our mathematical model is based on Euler-Newton equations to represent the 3D motion of the rigid body \cite{Sabatino2015,Beard2008}. To control the 6-DOF quadrotor system, a combination of translational $(x,y,z)$ and rotational $(\phi,\theta,\psi)$ motions is needed. From Newton’s law \cite{Sabatino2015}:
\begin{figure}[h]
\centering
\includegraphics[scale=0.4]{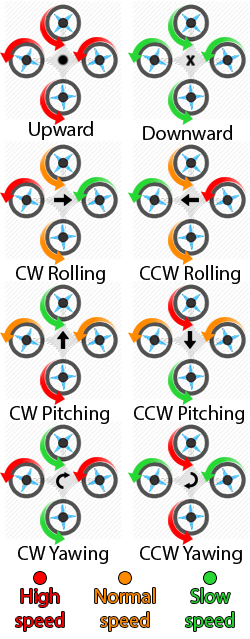}
\caption{Possible Movements of 6-DOF quadrotor system}
\label{figPossible Movements}
\end{figure}
\begin{figure}[h]
\centerline{\includegraphics[scale=1]{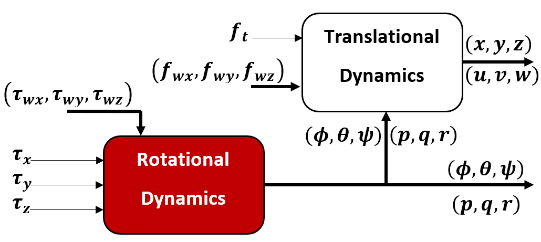}}
\caption{6-DOF Quadrotor system Dynamics Relations}
\label{figQuadrotor Dynamics Relations}
\end{figure}
\begin{figure*}[h]
\begin{equation}\label{eq3}
\left\{
\begin{array}{l}
\left(\!
    \begin{array}{c}
      \dot{x}\\
      \dot{y}\\
      \dot{z}
    \end{array}
\!\right) =
    \begin{pmatrix}
    c(\psi)c(\theta)&[c(\psi)s(\phi)s(\theta)-c(\phi)s(\psi)]&[s(\phi)s(\psi)+c(\phi)c(\psi)s(\theta)]\\
    c(\theta)s(\psi)&[c(\phi)c(\psi)+s(\phi)s(\psi)s(\theta)]&[c(\phi)s(\psi)s(\theta)-c(\psi)s(\phi)]\\
    -s(\theta)&c(\theta)s(\phi)&c(\phi)c(\theta))
    \end{pmatrix}
\left(\!
    \begin{array}{c}
      u\\
      v\\
      w
    \end{array}
\!\right)\\
\left(\!
    \begin{array}{c}
      \dot{u}\\
      \dot{v}\\
      \dot{w}
    \end{array}
\!\right)=
    \begin{pmatrix}
    0&r&-q\\
    -r&0&p\\
    q&-p&0    
    \end{pmatrix}
\left(\!
    \begin{array}{c}
      u\\
      v\\
      w
    \end{array}
\!\right)+
g\left(\!
    \begin{array}{c}
      -s(\theta)\\
      s(\phi)c(\theta)\\
      c(\phi)c(\theta)
    \end{array}
\!\right)+
\frac{1}{m}
\left(\!
    \begin{array}{c}
      f_{wx}\\
      f_{wy}\\
      f_{wz}-f_t
    \end{array}
\!\right)\\
\left(\!
    \begin{array}{c}
      \dot{\phi}\\
      \dot{\theta}\\
      \dot{\psi}
    \end{array}
\!\right)=
    \begin{pmatrix}
    1&s(\phi)t(\theta)&c(\phi)t(\theta)\\
    0&c(\phi)&-s(\phi)\\
    0&\frac{s(\phi)}{c(\theta)}&\frac{c(\phi)}{c(\theta)}
    \end{pmatrix}
\left(\!
    \begin{array}{c}
      p\\
      q\\
      r
    \end{array}
\!\right)\\
\left(\!
    \begin{array}{c}
      \dot{p}\\
      \dot{q}\\
      \dot{r}
    \end{array}
\!\right)=
\left(\!
    \begin{array}{c}
      \frac{I_y-I_z}{I_x}\\
      \frac{I_z-I_x}{I_y}\\
      \frac{I_x-I_y}{I_z}
    \end{array}
\!\right)
\left(\!
    \begin{array}{c}
      rq\\
      pr\\
      pq
    \end{array}
\!\right)+
\left(\!
    \begin{array}{c}
      \frac{\tau_x+\tau_{wx}}{I_x}\\
      \frac{\tau_y+\tau_{wy}}{I_y}\\
      \frac{\tau_z+\tau_{wz}}{I_z}
    \end{array}
\!\right)
\end{array}
\right.
\end{equation}
\end{figure*}
\begin{equation}\label{eq4}
\left(\!
    \begin{array}{c}
      \ddot{x} \\
      \ddot{y} \\
      \ddot{z}
    \end{array}
\!\right) =-\frac{f_t}{m}
    \left(\!
    \begin{array}{c}
      s(\phi)s(\psi)+c(\phi)c(\psi)s(\theta) \\
      c(\phi)s(\psi)s(\theta)-c(\psi)s(\phi) \\
      c(\phi)c(\theta)
    \end{array}
    \!\right)+
    \left(\!
    \begin{array}{c}
      0 \\
      0 \\
      g
    \end{array}
\!\right)
\end{equation}
\section{Problem statement}\label{Problem statement}
Suppose that the equations below represent the nonlinear 6-DOF quadrotor system shown in \autoref{fig quadrotor system}:
\begin{equation}\label{eq1}
\left\{
\begin{array}{l}
X^{n}= F(X,\dot X,...,X^{n-2},X^{n-1})+G(X)U\\
Y=X
\end{array}
\right.
\end{equation}
where $X=Y=[x,y,z,\phi,\theta,\psi,u,v,w,p,q,r]\in\mathbb{R}^{12}$ is the linear, angular position and velocity vectors of the quadrotor system, where $Y$ is the measured output, $U=[U_x,U_y,U_z,U_\phi,U_\theta,U_\psi]\in\mathbb{R}^6$ is the control input vector of the 6-DOF quadrotor system which needs to be designed such that it stabilizes the under actuated unstable 6-DOF quadrotor system and makes it follows a specific trajectory subject to optimum time-domain specifications and minimum control energy.
\begin{figure}[H]
\centerline{\includegraphics[scale=0.4]{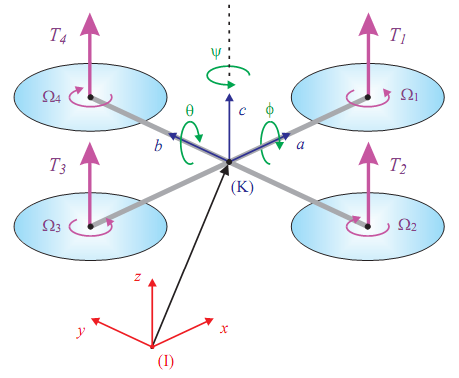}}
\caption{6-DOF quadrotor System}
\label{fig quadrotor system}
\end{figure}
\section{The main results}\label{Nonlinear Controller Design and stability Analysis}
\subsection{Nonlinear Controller Design}
Given the LPID controller, 
\begin{equation}\label{eq5}
U_{PID}=K_p\,e+K_d\,\dot{e}+K_i\int{e}\,dt
\end{equation}
which is a linear combination of the error signal $e\in\mathbb{R}$, its derivative $\dot{e}\in\mathbb{R}$, and its integral $\int{e}\,dt\in\mathbb{R}$ weighted by the gains $(K_p,K_d,K_i)\in\mathbb{R}^+$ respectively. The NLPID controller replaces each gain of the LPID controller with a nonlinear function $f(e)$ of the error signal so as to get a more satisfactory response for the nonlinear 6-DOF quadrotor system given as
\begin{equation}\label{eq6}
\!\left\{
\begin{array}{l}
U_{NLPID}=f_1(e)+f_2(\dot{e})+f_3(\int{e}\,dt) \\
f_i(\beta)=k_i(\beta)|\beta|^{\alpha_{i}}sign(\beta) \\
k_i(\beta)=k_{i1}+\frac{k_{i2}}{1+exp⁡(\mu_{i} \beta^2)},i={1,2,3}
\end{array}
\!\right.
\end{equation}
where $\beta$ could be $e$ , $\dot{e}$ , or $\int{e}\,dt$, $\alpha_{i}\in\mathbb{R}^+$, the function $k_i(\beta)$ is positive function with coefficients $k_{i1},k_{i2},\mu_{i}\in\mathbb{R}^+$. The nonlinear gain term $k_i(\beta)$ is bounded in the sector $[k_{i1}, k_{i1}+k_{i2}/2]$, see \autoref{figcharecter}. Less control energy is obtained with the NLPID control while the error is changing continuously. In this paper the NLPID controller is desighned as in \cite{Riyadh2017} with little modification in the integral term of \eqref{eq6}, where the term $k_{31}$ is added to increase the stability of the closed-loop system as will be shown later.
\begin{figure}[h]
\centering
\includegraphics[scale=0.35]{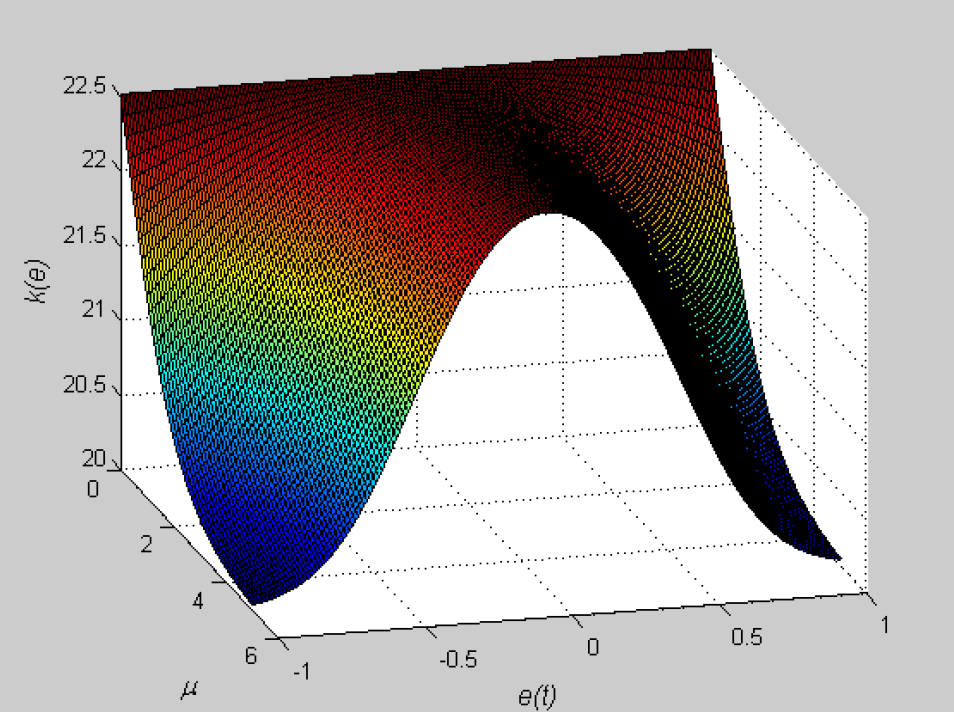}
\caption{Characteristics of the nonlinear gain function $k(e)$ for i=1,
$k_{11}=20, k_{12}=5$ \cite{Riyadh2017}}
\label{figcharecter}
\end{figure}\par
Because of the under actuated phenomenon of the 6-DOF quadrotor system, the control system for the 6-DOF quadrotor system is divided into two parts. The First part is where the input control is available, it is called the \textit{\textbf{Inner-Loop Control(ILC)}}, while the second part is where there is no actual input control available and it is called the \textit{\textbf{Outer-Loop Control(OLC)}}.
\subsubsection{Nonlinear design for ILC}
The proposed control signals for the altitude $z$ and the attitude ($\phi$, $\theta$, and $\psi$)  motions are given as follows, 
\par The throttle force control $f_t$ is given as
\begin{equation}\label{eq7}
\begin{array}{c}
U_z=f_t= f_{z1}(e_z)+f_{z2}(\dot{e}_z)+f_{z3}(\int{e_z}\,dt)\\
e_z=z_{de}-z_{ac}
\end{array}
\end{equation}
\par
While the Roll torque control signal $\tau_x$ is given as
\begin{equation}\label{eq8}
\begin{array}{c}
U_\phi=\tau_x= f_{\phi1}(e_\phi)+f_{\phi2}(\dot{e}_\phi)+f_{\phi3}(\int{e_\phi}\,dt)\\
e_\phi=\phi_{de}-\phi_{ac}
\end{array}
\end{equation}
\par
The Pitch torque control signal $\tau_y$ is given as
\begin{equation}\label{eq9}
\begin{array}{c}
U_\theta=\tau_y= f_{\theta1}(e_\theta)+f_{\theta2}(\dot{e}_\theta)+f_{\theta3}(\int{e_\theta}\,dt)\\
e_\theta=\theta_{de}-\theta_{ac}
\end{array}
\end{equation}
\par
Finally, the Yaw torque control $\tau_z$ can be designed as
\begin{equation}\label{eq10}
\begin{array}{c}
U_\psi=\tau_z= f_{\psi1}(e_\psi)+f_{\psi2}(\dot{e}_\psi)+f_{\psi3}(\int{e_\psi}\,dt)\\
e_\psi=\psi_{de}-\psi_{ac}
\end{array}
\end{equation}
where the subscript \enquote{$de$} means desired reference values, \enquote{$ac$} means the actual or measured values from different sensors of the 6-DOF quadrotor system, $f_{zi},f_{\phi i},f_{\theta i},f_{\psi i},i=1,2,3$ are the NLPID controller gains represented in \eqref{eq6}.
\subsubsection{Nonlinear design for OLC}
Quadrotor system has no real control input for the motion in the $(x,y)$ plane, the following analysis is proposed to generate the appropriate control signals ($U_x$ and $U_y$) for the motion in the $(x,y)$ plane, see \eqref{eq4}. Simplifying the model of \eqref{eq4} by assuming the $\phi$ and $\theta$ angle deviations are small (i.e. $\beta$ is small $\Rightarrow c(\beta)=1 ,s(\beta)=\beta$).
\begin{equation}\label{eq11}
\left(\!
    \begin{array}{c}
      \ddot{x} \\
      \ddot{y}
    \end{array}
\!\right) =\frac{U_z}{m}
    \begin{pmatrix}
      -s(\psi)&-c(\psi) \\
      -s(\psi)&c(\psi)
    \end{pmatrix}
    \left(\!
    \begin{array}{c}
      \phi_{de} \\
      \theta_{de}
    \end{array}
\!\right)
\end{equation}

\begin{figure*}[h]
\centering
\includegraphics[scale=0.5]{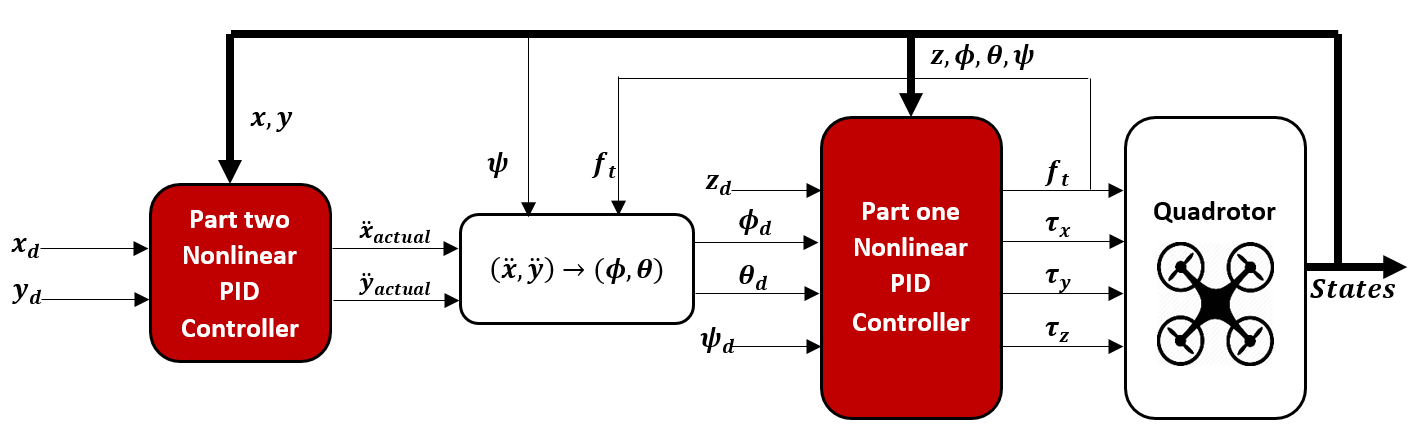}
\caption{Overall Quadrotor system}
\label{Overall 6-DOF Quadrotor System}
\end{figure*}
\begin{equation}\label{eq12}
\left(\!
    \begin{array}{c}
      \phi_{de} \\
      \theta_{de}
    \end{array}
\!\right) =\frac{m}{U_z}
    \begin{pmatrix}
      -s(\psi)&-c(\psi) \\
      -s(\psi)&c(\psi)
    \end{pmatrix}^{-1}
    \left(\!
    \begin{array}{c}
      \ddot{x} \\
      \ddot{y}
    \end{array}
\!\right)
\end{equation}\par
The main objective of the NLPID controller is to get the errors of $x$ and $y$ positions approach zero as described below
\begin{equation}\label{eq13}
\left\{\!
    \begin{array}{c}
      \ddot{x}_{de}-\ddot{x}_{ac}+f_{x1}(e_x)+f_{x2}(\dot{e}_x)+f_{x3}(\int{e}_x\,dt)\\=0 \\
      \ddot{y}_{de}-\ddot{y}_{ac}+f_{y1}(e_y )+f_{y2}(\dot{e}_y)+f_{y3}(\int{e}_y\,dt)\\=0
    \end{array}
\!\right.
\end{equation}
Substitute $\ddot{x}_{ac},\ddot{y}_{ac}$ of \eqref{eq13} into \eqref{eq12},yields,
\begin{equation}\label{eq14}
\left.\!
\begin{array}{c}
\left(\!
    \begin{array}{c}
      \phi_{de} \\
      \theta_{de}
    \end{array}
\!\right) =\frac{m}{U_z}
    \begin{pmatrix}
      -s(\psi)&-c(\psi) \\
      -s(\psi)&c(\psi)
    \end{pmatrix}^{-1}\times \\
    \left(\!
    \begin{array}{c}
	\ddot{x}_{de}+f_{x1}(e_x)+f_{x2}(\dot{e}_x)+f_{x3}(\int{e}_x\,dt) \\
    \ddot{y}_{de}+f_{y1}(e_y )+f_{y2}(\dot{e}_y)+f_{y3}(\int{e}_y\,dt)
    \end{array}
\!\right)
    \end{array}
\!\right.
\end{equation}
Letting $U_x, U_y$ be the virtual control signals for $x, y$ respectively, which are proposed as follows
\begin{equation}\label{eq14.0}
\begin{array}{c}
U_x=\ddot{x}_{ac}= f_{x1}(e_x)+f_{x2}(\dot{e}_x)+f_{x3}(\int{e_x}\,dt)\\
e_x=x_{de}-x_{ac}
\end{array}
\end{equation}
\begin{equation}\label{eq14.1}
\begin{array}{c}
U_y=\ddot{y}_{ac}= f_{y1}(e_y)+f_{y2}(\dot{e}_y)+f_{y3}(\int{e_y}\,dt)\\
e_y=y_{de}-y_{ac}
\end{array}
\end{equation}
where $f_{xi},f_{yi},i=1,2,3$ are the NLPID controller gains given in \eqref{eq6}.
By taking into consideration $\ddot{x}_{de}$ and $\ddot{y}_{de}=0$, then \eqref{eq14} can be written as,
\begin{equation}\label{eq15}
\left(\!
    \begin{array}{c}
      \phi_{de} \\
      \theta_{de}
    \end{array}
\!\right) =\frac{m}{U_z}
    \begin{pmatrix}
      -s(\psi)&c(\psi) \\
      -s(\psi)&-c(\psi)
    \end{pmatrix}
    \left(\!
    \begin{array}{c}
      U_x \\
      U_y
    \end{array}
\!\right)
\end{equation}\par
The Overall controlled quadrotor system is shown in \autoref{Overall 6-DOF Quadrotor System}.
\begin{remark}
To investigate the stability of the 6-DOF quadrotor system based on its nonlinear model, the nonlinear model equations of \eqref{eq3} is converted into six subsystems each one can be represented by a second order Brunovsky form given as
\begin{equation}\label{eq18.6}
\left\{\!
    \begin{array}{l}
      \dot{\zeta}_{1}=\zeta_{2} \\
      \dot{\zeta}_{2}=F_{U}+hU_{\zeta} \\
      \end{array}
\!\right.
\end{equation}
where $\zeta_1=\{x,y,z,\phi,\theta,\psi\}\in\mathbb{R}, \zeta_2=\{u,v,w,p,q,r\}\in\mathbb{R}$, $U_\zeta=\{U_x,U_y,-U_z,U_\phi,U_\theta,U_\psi\}\in\mathbb{R}$ and $h\in\mathbb{R}^+$, $F_U$  is unknown function which needs to be identified to analyze the stability of \eqref{eq18.6}. Some studies like \cite{Abdul-adheem2016} proposed an extended state observer to estimate the unknown function $F_U$, while \cite{Wang2006} proposed a radial basis function (RBF) to approximate $F_U$ and include it in the proposed control law. Since our proposed nonlinear PID controller belongs to a class of passive controllers, where the unknown function cannot be estimated by the controller itself, the stability test is based on the linearized model where the unknown function $F_U$ is diminished.
\end{remark}
\subsection{Stability Analysis of the closed-loop system}\label{Stability Analysis}
In this section, the overall stability analysis of both the \textbf{\textit{ILC}} and \textbf{\textit{OLC}} subsystems will be demonstrated using Routh-Hurwitz stability criterion on the linearized model of the 6-DOF quadrotor derived in \cite{Sabatino2015} with the virtual controllers $U_x$ and $U_y$ proposed in \eqref{eq14.0} and \eqref{eq14.1} respectively. The linearized model of the overall 6-DOF quadrotor system is represented as
\begin{equation}\label{eq15.5}
\left\{\!
    \begin{array}{l}
      \left.
      \begin{array}{l}
      \dot{x}_{1}=x_{2} \\
      \dot{x}_{2}=U_{\phi}/I_{x} \\
      \end{array}
      \quad\quad\right\}(19.a)
      \\
      \left.
      \begin{array}{l}
      \dot{x}_{3}=x_{4} \\
      \dot{x}_{4}=U_{\theta}/I_{y} \\
      \end{array}
      \quad\quad\,\right\}(19.b)
      \\
      \left.
      \begin{array}{l}
      \dot{x}_{5}=x_{6} \\
      \dot{x}_{6}=U_{\psi}/I_{z} \\
      \end{array}
      \qquad\right\}(19.c)
      \\
      \left.
      \begin{array}{l}
      \dot{x}_{7}=x_{8} \\
      \dot{x}_{8}=U_{x} \\
      \end{array}
      \qquad\quad\;\:\right\}(19.d)
      \\
      \left.
      \begin{array}{l}
      \dot{x}_{9}=x_{10} \\
      \dot{x}_{10}=U_{y} \\
      \end{array}
      \qquad\quad\,\right\}(19.e)
      \\
      \left.
      \begin{array}{l}
      \dot{x}_{11}=x_{12} \\
      \dot{x}_{12}=-U_{z}/m \\
      \end{array}
      \quad\right\}(19.f)
    \end{array}
\!\right.
\end{equation}
where $\textbf{x}\in\mathbb{R}^{12}:\textbf{x}=[\phi,p,\theta,q,\psi,r,x,u,y,v,z,w]$. 
Each one of the six subsystems described in \eqref{eq15.5} can be represented in the general form
\begin{equation}\label{eq15.51}
\left\{\!
    \begin{array}{l}
      \dot{\eta}_{1}=\eta_{2} \\
      \dot{\eta}_{2}=hU_{\eta} \\
      \end{array}
\!\right.
\end{equation}     
where $\eta_1,\eta_2\in\mathbb{R},\eta_1=x_i,\eta_2=x_{i+1}$ for $i=\{1,3,5,7,9,11\}$, $U_\eta=\{U_\phi,U_\theta,U_\psi,U_x,U_y,-U_z\}\in\mathbb{R}$ and $h\in\mathbb{R}^+$. If the stability of system \eqref{eq15.51} is ensured then the 6-DOF quadrotor subsystems are stable and the overall quadrotor system is stable in the sense of Hurwitz stability theorem. Before proceeding, some assumptions are needed.
\begin{assumption}\label{as1}
In order to prove the stability of the quadrotor system \eqref{eq15.5}, all $\alpha$,s in \eqref{eq6} will be approximated to $1$, since the tuned values of $\alpha$,s for the six controllers were almost around one, i.e. $\alpha_i\thickapprox1, i=1,2,...,6$ (i.e. $|\beta|\,sign(\beta)=\beta$).
\end{assumption}
\begin{assumption}\label{as2}
All the states $x,y,z,\phi,\theta$ and $\psi$ of the 6-DOF UAV quadrotor system are observable states, i.e. there is no need for the state observer. In case of non observable states, one can design an state observer or an extended state observer as in \cite{Ibraheem2018,Abdul-adheem2016}.
\end{assumption}
\begin{theorem}\label{th1}
Given the NLPID controller proposed in \Cref{eq7,eq8,eq9,eq10,eq14.0,eq14.1} for \eqref{eq15.5}, if assumptions \autoref{as1} and \autoref{as2} hold true, then the closed-loop system is Hurwitz stable for $k_1(e_1)\in[k_{11},k_{11}+\frac{k_{12}}{2}]$, $k_3(e_0)\in[k_{31},k_{31}+\frac{k_{32}}{2}]$, and $k_2(e_2)\in[k_{21},k_{21}+\frac{k_{22}}{2}]$.
\end{theorem}
\begin{proof}\label{proof1}
To proceed the proof of \autoref{th1}, each of the six subsystems of \eqref{eq15.5} can be represented by \eqref{eq15.51}, where the error dynamics for the closed-loop system can be written as
\begin{equation}\label{eq18.0}
\left\{\!
    \begin{array}{l}
      e_1=\eta_{1_{de}}-\eta_{1_{ac}}\\
      e_2=\eta_{2_{de}}-\eta_{2_{ac}}\\
      e_0=\int{e_1}\,dt
      \end{array}
\!\right.
\end{equation}
Taking the derivative of \eqref{eq18.0} and knowing that $\dot{\eta}_{1_{de}}=\eta_{2_{de}}$ and $\dot{\eta}_{2_{de}}=0$. This yields,
\begin{equation}\label{eq18.1}
\left\{\!
    \begin{array}{l}
      \dot{e}_0=e_1\\
      \dot{e}_1=e_2\\
      \dot{e}_2=-hU_{\eta}\\
      \end{array}
\!\right.
\end{equation}
Closing the loop by substituting the control signal \eqref{eq6} in \eqref{eq18.1} and expressing it in matrix form, we get,
\begin{equation}\label{eq18.2}
\begin{bmatrix}
\dot{e}_0  \\
\dot{e}_1  \\
\dot{e}_2
\end{bmatrix}=
A_c
\begin{bmatrix}
e_0  \\
e_1  \\
e_2
\end{bmatrix}
\end{equation}
where $A_c=\begin{bmatrix}
0&1&0  \\
0&0&1  \\
-hk_3(e_0)&-hk_1(e_1)&-hk_2(e_2)
\end{bmatrix}$
\par Finding the characteristics equation for $A_c$ by the relation $|\lambda I-A_c|$, yields,
\begin{equation}\label{eq23.0}
\lambda^3+hk_2(e_2)\lambda^2+hk_1(e_1)\lambda+hk_3(e_0)
\end{equation}\par
The Hurwitz Matrix for the characteristics equation is given as\\
$H=\begin{bmatrix}
hk_2(e_2)&hk_3(e_0)&0  \\
1&hk_1(e_1)&0 \\
0&hk_2(e_2)&hk_3(e_0)
\end{bmatrix}$\\
The conditions for the system \eqref{eq23.0} to be Hurwitz stable are given as\\
$\Delta_1=hk_2(e_2)>0$\\
$\Delta_2=h^2k_1(e_1)k_2(e_2)-hk_3(e_0)>0$\\
$\Delta_3=h^3k_1(e_1)k_2(e_2)k_3(e_0)-h^2k_3^2(e_0)>0$ \\
$\Delta_3=hk_3(e_0)\Delta_2>0$\\
As we stated before, $k_i(\beta)$ is sector bounded in the range $[k_{i1}, k_{i1}+k_{i2}/2]$ and always positive, assuming a range for any two of $\{k_1(e_1),k_2(e_2),k_3(e_0)\}$ will lead to the range of the third one , e.g., let $k_1(e_1)\in[k_{11},k_{11}+\frac{k_{12}}{2}]$ and $k_3(e_0)\in[k_{31},k_{31}+\frac{k_{32}}{2}]$, this will results in $k_2(e_2)>k_3(e_0)/h{k_1(e_1)}$ and gives $k_{21}>\frac{k_{31}}{hk_{11}+h\frac{k_{12}}{2}}$ and $k_{21}+\frac{k_{22}}{2}>\frac{k_{31}+\frac{k_{32}}{2}}{hk_{11}}$, which will ensure the closed-loop system to be stable in the sense of the Hurwitz stability theorem.
\end{proof}
\begin{remark}
The linearized model \eqref{eq15.5} of the 6-DOF quadrotor is used only to prove the stability of the closed-loop quadrotor system, while the complete nonlinear model represented by \eqref{eq3} is used in the control design and simulations.
\end{remark}
\section{Simulation Results and  Case Studies}\label{Simulation and Results}
\subsection{Step reference tracking}
The 6-DOF nonlinear quadrotor dynamic model and the NLPID controller are implemented in MATLAB/Simulink, where we have assumed that the wind forces and torques $[f_{wx},f_{wy},f_{wz},\tau_{wx},\tau_{wy},\tau_{wz}]$ are negligible. The
parameters values of the quadrotor used in the simulations are listed in \autoref{table2}. In our simulations an unconstrained multi-objective optimization is conducted to tune the parameters of the NLPID controllers, i.e.,\par 
\begin{center}
Find $\Gamma=[\Gamma_{x},\Gamma_{y},\Gamma_{z},\Gamma_{\phi},\Gamma_{\theta},\Gamma_{\psi}]$ which minimizes
\begin{equation}\label{eq19}
\left\{\!
\begin{array}{l}
opi_i=w_{1i} \times \frac{ITAE_i}{N_{1i}}+w_{2i} \times \frac{USQR_i}{N_{2i}}\\
OPI=\sum_i (\widehat{w}_{i} \times iop_i), i=x,y,z,\psi
\end{array}  
\!\right. 
\end{equation}
where $\Gamma_\sigma=[k_{11}\;k_{12}\;k_{21}\;k_{22}\;k_{31}\;k_{32}\;\mu_1\;\mu_2\;\mu_3\;\alpha_1\;\alpha_2\;\alpha_3]$,\\ $\sigma=\{x,y,z,\phi,\theta,\psi\}$
\end{center}
while for the LPID controller 
\begin{center}
Find $\Gamma=[\Gamma_{x},\Gamma_{y},\Gamma_{z},\Gamma_{\phi},\Gamma_{\theta},\Gamma_{\psi}]$ which minimizes \Cref{eq19}
where $\Gamma_\sigma=[Kp\;Kd\;Ki]$, $\sigma=\{x,y,z,\phi,\theta,\psi\}$
\end{center}

\begin{center}
\captionof{table}{Parameters’ Values} \label{table2} 
\scalebox{1}{
  \begin{tabular}{ c | c }
    \hline
   Parameter&Value \\ \hline
   $I_x$&$8.5532 \times 10^{-3}$\\ \hline
   $I_y$&$8.5532\times 10^{-3}$\\ \hline
   $I_z$&$1.476\times 10^{-2}$\\ \hline
   $g$&$9.81$\\ \hline
   $m$&$0.964$\\ \hline
   $b$&$7.66\times 10^{-5}$\\ \hline
   $d$&$5.63\times 10^{-6}$\\ \hline
   $l$&$0.22$\\
   \hline
  \end{tabular}
  }
\end{center}
\par
The ITAE and USQR are defined in \autoref{table3}. The weighting variables must satisfy $w_{1i}+w_{2i}=1$, They are defined as the relative emphasis of one objective as compared to the other. Their values are chosen to increase the pressure on selected objective function. The same applies for the weighting variables $\widehat{w}_i$. While $N_{1i}$ and $N_{2i}$ are normalizing variables. They are included in the performance index to ensure that the individual objectives  have comparable values, and are treated equally likely by the tuning algorithm. Because, if a certain objective is of very high value, while the second one has very low value, then the tuning algorithm will pay much consideration to the highest one and leave the other with little reflection on the system. The reason for taking just $x,y,z$, and $\psi$ in the optimization process is due to the fact that the position of UAV is characterized by the 3D coordinates (\textit{x,y,z}) and rotation about z-axis $\psi$. Adding $\phi$ and $\theta$ will result in a better performance but the optimization process will take longer time. The values of $w_{1i}=0.6, w_{2i}=0.4, \widehat{w}_i=0.25, i={x,y,z,\psi}$, $N_{1i}=N_{2i}=1$ for $i={x,y,\psi}$, while $N_{1z}=1$ and $N_{2z}=4500$. The values of both controllers after tuning are listed in \autoref{table4} and \autoref{table5}. After tuning the overall quadrotor’s six controllers, the performance indices for ($x,y,z,$ and $\psi$) and the 6-DOF quadrotor total $OPI$ are shown in \autoref{table11}.
\begin{center}
\captionof{table}{Performance Indices} \label{table3} 
  \scalebox{1}{
  \begin{tabular}{ P{2cm} | P{2.5cm} | P{2cm} }
    \hline 
   Performance Index& \qquad\quad Description&Mathematical Representation \\ \hline 
   ITAE&Integrated time absolute error&$\int_0^{t_f}{t |e(t)|}\,dt$\\ \hline
   USQR&Controller energy&$\int_0^{t_f}{[u(t)]^2}\,dt$\\
   \hline
  \end{tabular}
  }\par
  *$t_f$ is the final time of simulation
\end{center}
\begin{table}[H]
\begin{center}
\captionof{table}{ LPID Parameters} \label{table4} 
\scalebox{1}{
  \begin{tabular}{ c  c | c | c }
     \cline{2-4}
   	&$k_p$&$k_i$&$k_d$ \\ \hline 
   	\multicolumn{1}{ c| }{$x$}&$0.28$&$2.73\times10^{-6}$&$0.63$\\ \hline
	\multicolumn{1}{ c| }{$y$}&$0.36$&$1.56\times10^{-5}$&$0.88$\\ \hline
	\multicolumn{1}{ c| }{$z$}&$184.02$&$103.73$&$22.5$\\ \hline
	\multicolumn{1}{ c| }{$\phi$}&$0.88$&$0.9$&$0.3$\\ \hline
	\multicolumn{1}{ c| }{$\theta$}&$0.62$&$0.81$&$0.05$\\ \hline
	\multicolumn{1}{ c| }{$\psi$}&$0.99$&$0.49$&$0.56$\\
   \hline
  \end{tabular}
  }
\end{center}
\end{table}
\begin{center}
\captionof{table}{ NLPID Parameters} \label{table5} 
\scalebox{0.9}{
  \begin{tabular}{ c  P{0.7cm} | P{0.7cm} | c | c | c | c }
     \cline{2-7}
      &$x$&$y$&$z$&$\phi$&$\theta$&$\psi$\\ \hline 
      \multicolumn{1}{ c| }{$k_{11}$}&$	1.51$&$	1.38$&$	27.5$&$	0.77$&$	0.48$&$	0.76$\\ \hline
	  \multicolumn{1}{ c| }{$k_{12}$}&$	0.04$&$	0.03$&$	8.76$&$	0.06$&$	0.03$&$	0.16$\\ \hline
	  \multicolumn{1}{ c| }{$k_{21}$}&$	1.13$&$	2.51$&$	8.8$&$	0.2$&$	0.08$&$	0.17$\\ \hline
	  \multicolumn{1}{ c| }{$k_{22}$}&$	0.18$&$	0.04$&$	4.71$&$	0.04$&$	0.12$&$	0.11$\\ \hline
	  \multicolumn{1}{ c| }{$k_{31}$}&$	1.81\times10^{-6}$&$	5.72\times10^{-5}$&$	18.49$&$	1.08$&$	0.88$&$	0.27$\\ \hline
	  \multicolumn{1}{ c| }{$k_{32}$}&$	10^{-6}$&$	8.69\times10^{-6}$&$	10.02$&$	0.08$&$	0.11$&$	0.08$\\ \hline
	  \multicolumn{1}{ c| }{$\mu_1$}&$	0.11$&$	0.08$&$	0.31$&$	0.07$&$	0.84$&$	0.25$\\ \hline
	  \multicolumn{1}{ c| }{$\mu_2$}&$	0.08$&$	0.36$&$	0.36$&$	0.56$&$	1.43$&$	0.46$\\ \hline
	  \multicolumn{1}{ c| }{$\mu_3$}&$	0.18$&$	0.6$&$	0.98$&$	0.58$&$	0.28$&$	0.81$\\ \hline
	  \multicolumn{1}{ c| }{$\alpha_1$}&$	0.93$&$	0.93$&$	0.96$&$	0.96$&$	0.96$&$	0.98$\\ \hline
	  \multicolumn{1}{ c| }{$\alpha_2$}&$	0.93$&$	0.92$&$	0.97$&$	0.96$&$	1$&$	0.95$\\ \hline
	  \multicolumn{1}{ c| }{$\alpha_3$}&$	0.95$&$	0.93$&$	0.97$&$	0.97$&$	0.97$&$	0.92$\\
   \hline
  \end{tabular}
  }
\end{center}
\begin{center}
\captionof{table}{ Position and Yaw Performance indices} \label{table11} 
\scalebox{0.9}{
  \begin{tabular}{ c  c | c | c | c  }
     \cline{2-5}
      &\multicolumn{2}{ c| }{LPID}&\multicolumn{2}{ c }{NLPID}\\ \cline{2-5}
      &ITAE&ISU&ITAE&ISU\\ \hline
      \multicolumn{1}{ c| }{$x$}&$14.285931$&$0.134411$&$0.438883$&$0.381936$\\ \hline
	  \multicolumn{1}{ c| }{$y$}&$7.498694$&$0.323482$&$1.066173$&$1.226223$\\
   \hline
   	  \multicolumn{1}{ c| }{$z$}&$0.059225$&$5197.496$&$0.152148$&$4516.303$\\
   \hline
   	  \multicolumn{1}{ c| }{$\psi$}&$1.377493$&$0.030779$&$0.506028$&$0.037517$\\
   \hline
   \multicolumn{1}{ c| }{$OPI$}&\multicolumn{2}{ c| }{$3.6476$}&\multicolumn{2}{ c }{$0.5894$}\\ \hline
  \end{tabular}
  }
\end{center}
\par 
A unit step reference inputs ($x_{de}$, $y_{de}$, $z_{de}$, and $\psi_{de}$) have been applied to the position ($x,y,$ and $z$) of the 6-DOF and the yaw ($\psi$) orientation. the curves in \autoref{x-position} (a), represent the response of the $x$-position of the quadrotor system using both controllers, while (b) represent the energy signal produced by the controllers to achieve the required position. The output response for $y, z,$ and $\psi$ are shown in \Cref{y-position,z-position,ψ-position}.  As can be seen, the control signal produced by the NLPID controller is less fluctuating than that of the LPID controller. The NLPID controller shows a faster  response than the LPID one, except for the $z$-position, where the LPID controller presents a faster tracking, but on the account of a large control energy being spent. This increase in the energy of the control signal is undesired in the practice, since it leads to actuator saturation. The overshoot in the output response for the PID controller is very clear. The output responses of $\phi$ and $\theta$ are drawn in \Cref{ϕ-position,θ-position} respectively which show the effectiveness of the proposed NLPID controller over the linear one.
 \begin{figure}[h]
    \centering
    \begin{subfigure}{0.7\linewidth}
        \includegraphics[width=\linewidth]{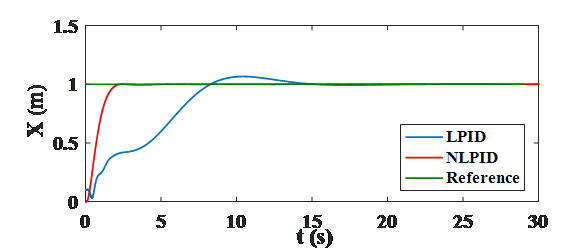}
        \caption{}
    \end{subfigure}
    \begin{subfigure}{0.7\linewidth}
        \includegraphics[width=\linewidth]{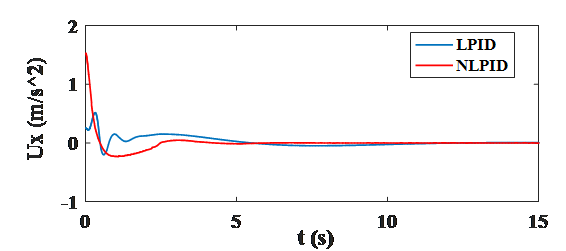}
        \caption{}
    \end{subfigure}
    \caption{$x$-position (a) time response (b) controller signal}
    \label{x-position}
\end{figure}
\begin{figure}[!]
    \centering
    \begin{subfigure}{0.7\linewidth}
        \includegraphics[width=\linewidth]{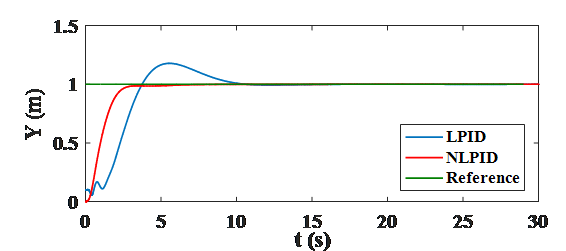}
        \caption{}
    \end{subfigure}
    \begin{subfigure}{0.7\linewidth}
        \includegraphics[width=\linewidth]{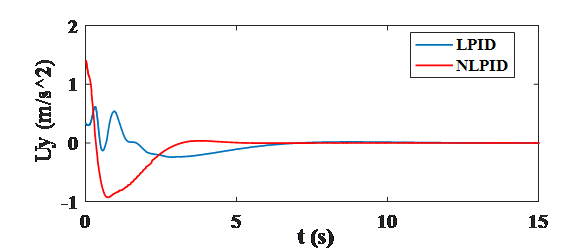}
        \caption{}
    \end{subfigure}
    \caption{$y$-position (a) time response (b) controller signal}
    \label{y-position}
\end{figure}
\begin{figure}[H]
    \centering
    \begin{subfigure}{0.7\linewidth}
        \includegraphics[width=\linewidth]{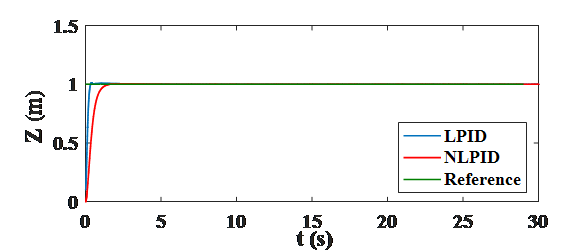}
        \caption{}
    \end{subfigure}
    \begin{subfigure}{0.7\linewidth}
        \includegraphics[width=\linewidth]{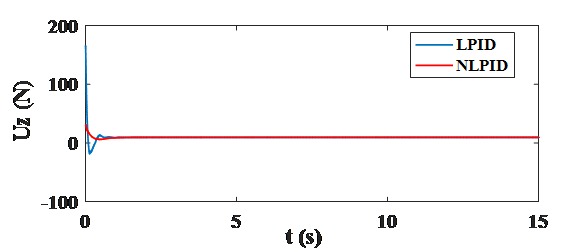}
        \caption{}
    \end{subfigure}
    \caption{$z$-position (a) time response (b) controller signal}
    \label{z-position}
\end{figure}
\begin{figure}[h]
    \centering
    \begin{subfigure}{0.7\linewidth}
        \includegraphics[width=\linewidth]{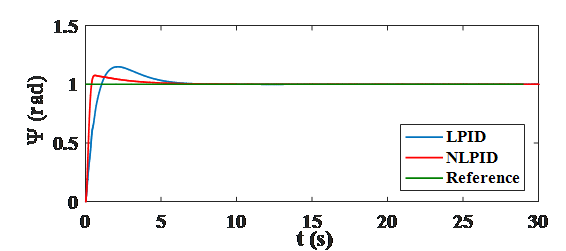}
        \caption{}
    \end{subfigure}
    \begin{subfigure}{0.7\linewidth}
        \includegraphics[width=\linewidth]{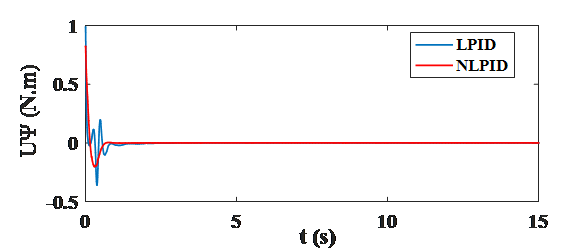}
        \caption{}
    \end{subfigure}
    \caption{$\psi$-position (a) time response (b) controller signal}
    \label{ψ-position}
\end{figure}
\begin{figure}[H]
    \centering
    \begin{subfigure}{0.7\linewidth}
        \includegraphics[width=\linewidth]{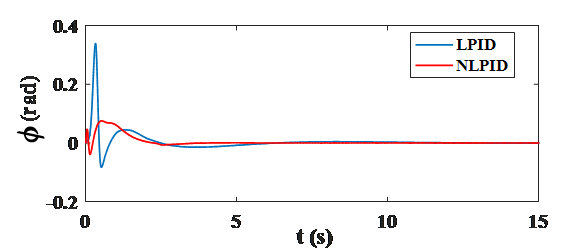}
        \caption{}
    \end{subfigure}
    \begin{subfigure}{0.7\linewidth}
        \includegraphics[width=\linewidth]{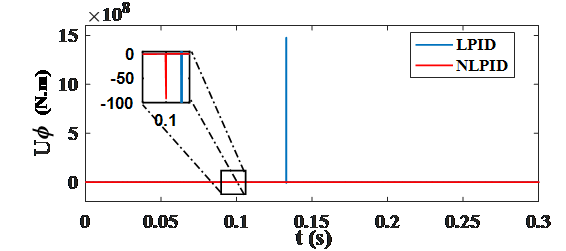}
        \caption{}
    \end{subfigure}
    \caption{$\phi$-position (a) time response (b) controller signal}
    \label{ϕ-position}
\end{figure}
\begin{figure}[H]
    \centering
    \begin{subfigure}{0.7\linewidth}
        \includegraphics[width=\linewidth]{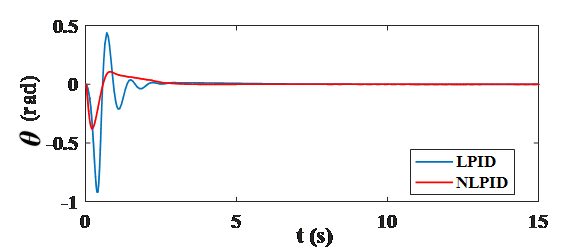}
        \caption{}
    \end{subfigure}
    \begin{subfigure}{0.7\linewidth}
        \includegraphics[width=\linewidth]{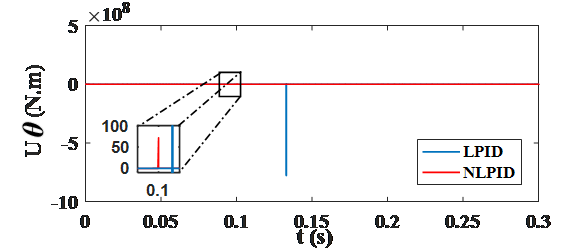}
        \caption{}
    \end{subfigure}
    \caption{$\theta$-position (a) time response (b) controller signal}
    \label{θ-position}
\end{figure}
This has been definitely reflected on the energy of the control signal and the smoothness of the output response. Comparing the energy signals for $\phi$ and $\theta$ using both controllers, the NLPID controller presented a hug reduction as compared to the linear one. The time-domain specifications of both controllers are presented numerically in \autoref{table6} (rising $t_r$, settling $t_s$ and peak overshot $M_p$) for the position ($x$, $y$, and $z$) of the quadrotor system and yaw($\psi$) using both controllers.\autoref{table7} shows the minimum and maximum peaks for $\phi$ and $\theta$ using both controllers.
 \begin{center}
\captionof{table}{ Position and Yaw responses} \label{table6} 
\scalebox{0.9}{
  \begin{tabular}{ c  c | c | c | c | c | c }
     \cline{2-7}
      &\multicolumn{3}{ c| }{LPID}&\multicolumn{3}{ c }{NLPID}\\ \cline{2-7}
      &$t_r(s)$&$t_s(s)$&$M_p\%$&$t_r(s)$&$t_s(s)$&$M_p\%$\\ \hline
      \multicolumn{1}{ c| }{$x$}&$6.546$&$13.679$&$7.471$&$1.152$&$4.657$&$0.505$\\\hline
	  \multicolumn{1}{ c| }{$y$}&$1.934$&$9.748$&$19.375$&$1.572$&$3.023$&$0.195$\\\hline
	  \multicolumn{1}{ c| }{$z$}&$0.194$&$0.314$&$1.531$&$0.677$&$1.283$&$0.505$\\\hline
	  \multicolumn{1}{ c| }{$\psi$}&$0.681$&$5.660$&$19.014$&$0.252$&$3.840$&$8.152$\\
   \hline
  \end{tabular}
  }
\end{center}
\begin{center}
\captionof{table}{ Roll and Pitch responses} \label{table7} 

  \begin{tabular}{ c  P{1.5cm} | P{1cm} | P{1.5cm} | P{1cm}  }
     \cline{2-5}
      &\multicolumn{2}{ c| }{LPID}&\multicolumn{2}{ c }{NLPID}\\ \cline{2-5}
      &Min peak&	Max peak	&Min peak	&Max peak\\ \hline
      \multicolumn{1}{ c| }{$\theta$}&$-8.292\times10^{-2}$&$3.91\times10^{-1}$&$	-3.951\times10^{-2}$&$7.402\times10^{-2}$\\ \hline
	  \multicolumn{1}{ c| }{$\psi$}&$-9.195\times10^{-1}$&$4.385\times10^{-1}$&$	-3.775\times10^{-1}$&$1.065\times10^{-1}$\\
   \hline
  \end{tabular}
\end{center}

\subsection{Trajectory tracking}\label{Case studies}
\quad In the following, the output responses of three case studies with trajectories generated by MATLAB/Simulink to test the overall 6-DOF UAV system using linear and nonlinear PID controllers will be presented and discussed. The cases studied have been chosen to reflect the difficulties that the quadcopter control system might face in achieving the required tracking. For all of the case studies, the initial values are $x=0.1,y=0.1,z=0.1$ and the rest of the states are zero.
\paragraph{Case Study (1)}\label{Case1}
The first test case study was the circular path. The States with the reference trajectories are presented in \autoref{table8}.
\begin{center}
\captionof{table}{Case 1 input signals} \label{table8} 
  \begin{tabular}{ c | c | c }
    \hline 
   State&	Reference Trajectory&	Time/sec\\ \hline
	$x$&$	cos⁡(0.1\pi t)$&$5-t_f$\\ \hline
	$y$&$	sin⁡(0.1\pi t)$&$5-t_f$\\ \hline
	$z$&$	u(t)$      &$0-t_f$\\ \hline
	$\psi$&$	u(t)$      &$0-t_f$\\
   \hline
  \end{tabular}\par
  *$t_f=50sec$
\end{center}
\par \autoref{fig-c1} shows the tracking of the 6-DOF UAV system for the circular trajectory. The proposed NLPID controller followed the trajectory with less time and smaller error than in the LPID one. The steady state error in the LPID controller was $e_x=27\%,e_y=18\%,$ and $e_z=0\%$ , while in the proposed NLPID controller, the steady state error was $e_x=0.56\%,e_y=4.12\%,$ and $e_z=0\%$.
\begin{figure}[H]
\centerline{\includegraphics[scale=0.5]{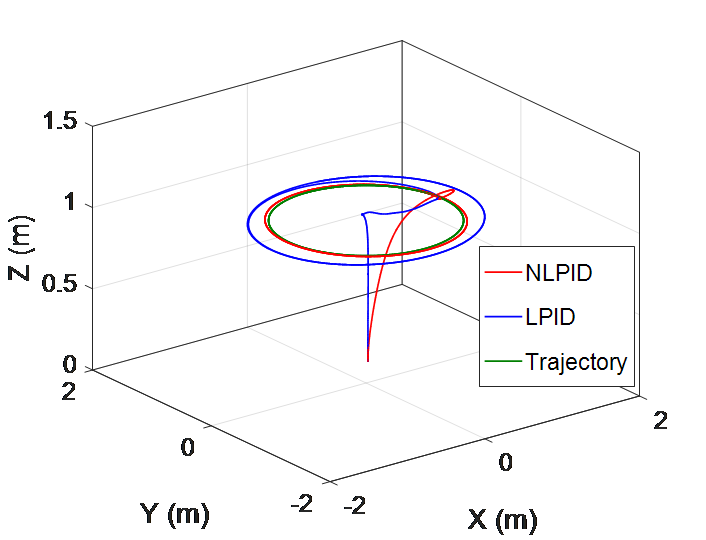}}
\caption{Case 1 - Circular Trajectory}
\label{fig-c1}
\end{figure}
\paragraph{Case Study (2)}\label{Case2}
The Second test case was the helical path, in this case study the altitude $z$ of the 6-DOF UAV is varying with time, in contrast to the first case study. The states with reference trajectories are presented in \autoref{table9}.
\begin{center}
\captionof{table}{Case 2 input signals} \label{table9} 
  \begin{tabular}{ c | c | c }
    \hline 
   State&	Reference Trajectory&	Time/sec\\ \hline
	$x$&$	cos⁡(0.1\pi t)$&$5-t_f$\\ \hline
	$y$&$	sin⁡(0.1\pi t)$&$5-t_f$\\ \hline
	$z$&$	0.2\,t$      &$0-t_f$\\ \hline
	$\psi$&$	u(t)$      &$0-t_f$\\
   \hline
  \end{tabular}\par
  *$t_f=100sec$
\end{center}
\par
The LPID controller follows the desired trajectory with constant offset for the entire time of the simulations. While the proposed NLPID controller showed an improved performance over the LPID one as shown in \autoref{fig-c2}.
\begin{figure}[h]
\centerline{\includegraphics[scale=0.5]{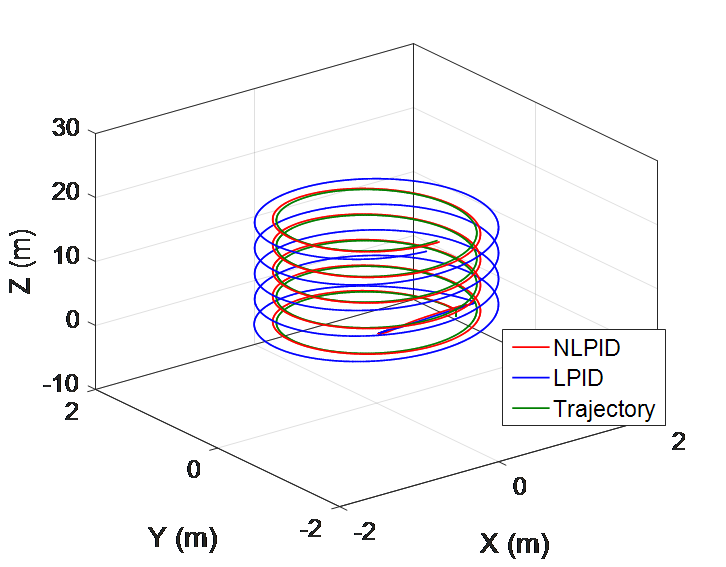}}
\caption{Case 2 - Helical Trajectory}
\label{fig-c2}
\end{figure}
\paragraph{Case Study (3)}\label{Case3}
The final case study is the square trajectory. This trajectory is a serious test for the controllers designed for the UAV systems to accomplish the required track, since, the trajectory changes its direction suddenly at certain times (i.e., at the vertices of the  square). The states with reference trajectories are presented in \autoref{table10}.
\begin{center}
\captionof{table}{Case 3 input signals} \label{table10} 
  \begin{tabular}{ c | c | c }
    \hline 
   State&	Reference Trajectory&	Time/sec\\ \hline
	$x$&$	u(t-10)-u(t-50)$&$10-50$\\ \hline
	$y$&$	(t-30)-u(t-70)$&$30-70$\\ \hline
	$z$&$	u(t)$      &$0-t_f$\\ \hline
	$\psi$&$	u(t)$      &$0-t_f$\\
   \hline
  \end{tabular}\par
  *$t_f=100sec$
\end{center}
We can see the difference between the responses of both controllers in \autoref{fig-c3}. The significance of the NLPID controller is very obvious, it tracked faster than the LPID one with very small overshot as compared to the LPID controller. The overshot in the LPID reached approximately up to $200\%$ of the desired trajectory.
\begin{figure}[h]
\centerline{\includegraphics[scale=0.46]{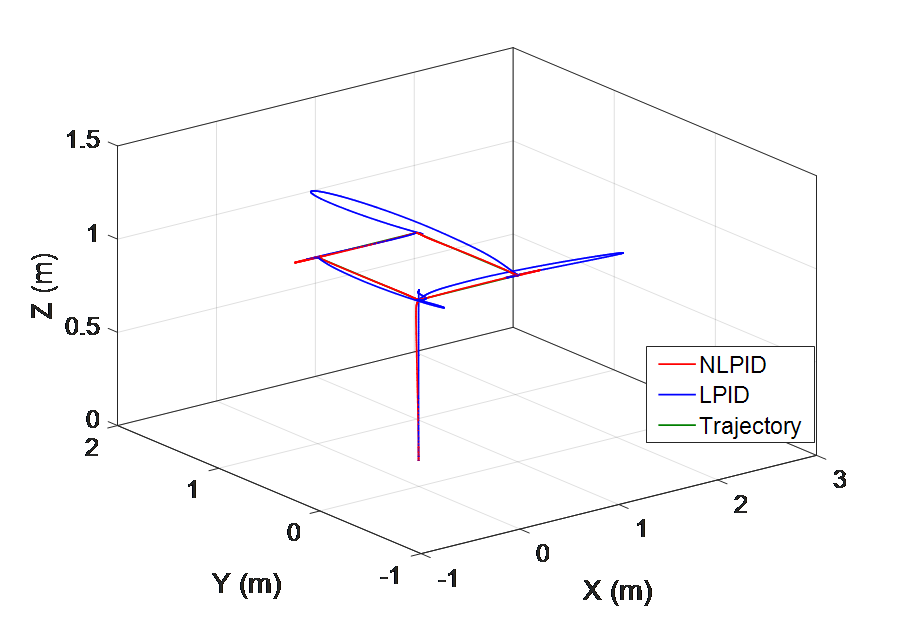}}
\caption{Case 3 - Square Trajectory}
\label{fig-c3}
\end{figure}
\section{Conclusion and Future Work}\label{Conclusion and Future Work}
\quad In this work, the exact nonlinear model of the 6-DOF UAV was adopted with the aim of designing a nonlinear controller for the stabilization and reference tracking for this complex and highly nonlinear system.  A new NLPID controller for stabilizing and controlling a 6-DOF UAV system is proposed in this paper. The stability analysis for the closed-loop systems of both the position and orientation of the 6-DOF UAV system based on Hurwitz stability criterion was analyzed and proved that the proposed controller stabilized the system and accomplished the required tracking. From the simulation tests, it can be concluded that the proposed NLPID controller is better than the linear one in terms of speed,  control energy, and steady state error. For future work, we suggest taking wind disturbances into account and rejecting them by designing a disturbance observer using active disturbance rejection control paradigm for this purpose.
\section*{Acknowledgment}
\quad A special thanks go to the Ph.D. student Wameedh Riyadh Abdul-Adheem for the deep discussions on the topics of this paper.\\
\bibliography{library}

\end{document}